\documentclass[conference]{IEEEtran}
\usepackage{amsfonts}
\usepackage{amssymb}
\usepackage{cite}
\usepackage[cmex10]{amsmath}
\usepackage{float}
\usepackage{color}
\usepackage{stfloats,fancyhdr}
\usepackage{amsmath}
\usepackage{algorithm}
\usepackage{algorithmic}
\usepackage{multirow}
\usepackage{changepage}
\usepackage[normalem]{ulem}

\newtheorem{proposition}{Proposition}

\newtheorem{remark}{Remark}

\IEEEoverridecommandlockouts

\ifCLASSINFOpdf
  \usepackage[pdftex]{graphicx}
  \DeclareGraphicsExtensions{.pdf,.jpeg,.png}
\else
  \usepackage[dvips]{graphicx}
  \DeclareGraphicsExtensions{.eps}
\fi

\usepackage{subfigure}

\usepackage{fancybox,dashbox}

\begin{document}
%
\title{An Adaptive Transmission Protocol for Wireless-Powered Cooperative Communications}
\author{Yifan Gu, He~(Henry)~Chen, Yonghui Li, and Branka Vucetic
\\
School of Electrical and Information Engineering, The University of Sydney, Sydney, NSW 2006, Australia\\
Email: yigu6254@uni.sydney.edu.au, \{he.chen, yonghui.li, branka.vucetic\}@sydney.edu.au
}

\maketitle

\begin{abstract}
In this paper, we consider a wireless-powered cooperative communication network, which consists of one hybrid access point (AP), one source and one relay to assist information transmission. Unlike conventional cooperative networks, the source and relay are assumed to have no embedded energy supplies in the considered system. Hence, they need to first harvest energy from the radio-frequency (RF) signals radiated by the AP in the downlink (DL) before information transmission in the uplink (UL). Inspired by the recently proposed harvest-then-transmit (HTT) and harvest-then-cooperate (HTC) protocols, we develop a new adaptive transmission (AT) protocol. In the proposed protocol, at the beginning of each transmission block, the AP charges the source. AP and source then perform channel estimation to acquire the channel state information (CSI) between them. Based on the CSI estimate, the AP adaptively chooses the source to perform UL information transmission either directly or cooperatively with the relay. We derive an approximate closed-form expression for the average throughput of the proposed AT protocol over Nakagami-m fading channels. The analysis is then verified by Monte Carlo simulations. Results show that the proposed AT protocol considerably outperforms both the HTT and HTC protocols.
\end{abstract}


\IEEEpeerreviewmaketitle
\section{Introduction}
Conventional energy harvesting techniques harvest energy from external sources such as solar power, wind energy, etc\cite{background1,background2}. However, these power sources are not part of the communication network and may introduce extra complexity to the system. Recently, a novel RF energy harvesting technique has been proposed that the device can collect energy from the radio frequency (RF) signals \cite{background3,background4}. A new kind of networks exploiting RF energy harvesting, named wireless-powered communication networks (WPCNs), have emerged and drawn more and more attention. In a WPCN, wireless devices are only powered by the RF signals. In \cite{harvest-then-transmit}, a WPCN consists of one hybrid access point (AP) and a set of users with no embedded energy supplies was investigated. A harvest-then-transmit (HTT) protocol was proposed. In HTT, the users harvest energy from the AP in the downlink (DL) before transmitting information in the uplink (UL).

\cite{chen2014harvest} developed a framework of wireless-powered cooperative communication networks (WPCCNs), where some energy harvesting relays are deployed to improve the performance of WPCCNs.  A typical three-node WPCCN consists of one hybrid AP, one information source (S) and one relay (R) to assist information transmission. Based on this three-node WPCCN, a harvest-then-cooperate (HTC) protocol was developed in \cite{chen2014harvest}, in which the source and relay harvest energy in the DL at the same time and work cooperatively for UL information transmission. Considering a delay-limited transmission model and selection combining technique at the AP, the average throughput of the WPCCNs was analyzed 
by evaluating the system outage probability with a fixed transmission rate at the source. It is shown in [6] that the HTC protocol can yield considerable performance gain compared to the HTT protocol, by introducing the cooperation between the source and relay. Almost at the same time, similar cooperative schemes were proposed for WPCNs with different setups in \cite{ju2014user,chen2014wireless}.

However, when we compare the HTT with HTC protocols for instantaneous channel realizations, we find that for some channel realizations, outages occur for the HTC protocol, but not for the HTT protocol. This means that the HTC does not always outperform the HTT in all channel realizations. Motivated by this, we develop an adaptive transmission (AT) protocol, where the AP adaptively chooses the source to perform UL information transmission either directly or cooperatively with the relay. To make this decision, the AP only needs to know the channel information  between itself and the source, which can be estimated at the beginning of each transmission block.

The main contributions of this paper are summarized as follows: {\textbf{(1)}} We propose an adaptive transmission protocol for the WPCCNs, in which the HTT and HTC protocols are adaptively adopted based on the channel information\footnote{Here, we consider the case that at the beginning of each transmission block, the AP first transfers a certain amount of energy to the source that is dedicated to acquire the channel information between them. After this duration, the AP delivers wireless energy again, which will be used for UL information transmission. For the purpose of exploration, we assume that the time duration for this channel information acquisition is very short such that it can be ignored compared with the whole block.} between the AP and the source. {\textbf{(2)}} Considering the delay-limited transmission scheme \cite{Delay-limitedthroughput} and the amplify-and-forward protocol \cite{SC} employed at the relay, we derive an approximate closed-form expression for the average throughput of the proposed protocol by evaluating the outage probability over Nakagami-m fading channels. {\textbf{(3)}} To perform the performance comparison, we also analyze the average throughput performance of the HTT and HTC protocols over Nakagami-m fading channels. {\textbf{(4)}} Numerical simulations are performed to validate the analytical results. Results show that the proposed AT protocol considerably outperforms both the HTT and HTC protocols.
\begin{figure}
\centering \scalebox{0.5}{\includegraphics{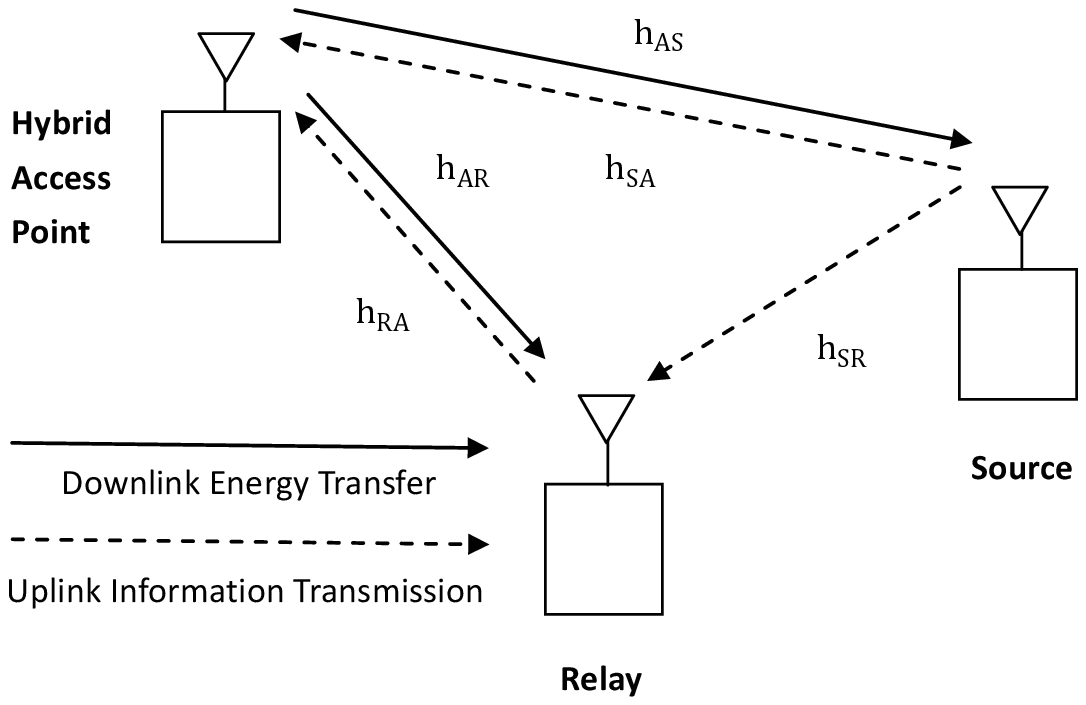}}
\caption{System model of a WPCCN. \label{fig:WPCCN}}
\end{figure}

\section{System Model}
As shown in Fig.\ref{fig:WPCCN}, we consider a three-node WPCCN with wireless energy transfer (WET) in the downlink (DL) and wireless information transmission (WIT) in the uplink (UL). The network consists of one hybrid AP, one information source (S) and one relay (R) to assist information transmission. The hybrid AP is connected to external power supplies while the source and relay are assumed to have no energy supplies. However, their batteries could store the energy harvested from the RF signals broadcast by the AP. Throughout the paper, we use subscript-$A$ for hybrid AP, subscript-$S$ for source and subscript-$R$ for relay. We consider that the channel between any two nodes $X$ and $Y$ suffers from Nakagami-m fading with fading severity parameter\footnote{For the purpose of exploration, we consider that all the fading severity parameters are integers in this paper.} ${m_{XY}}$ and average power ${\Omega _{XY}}$, where $X,Y \in \left\{ {A,S,R} \right\}$. Note that the Nakagami-m fading is an important channel model that can capture the physical channel phenomena more accurately than Rayleigh and Rician models. Moreover, we use ${h_{XY}}$ to denote the channel power gain from node $X$ to node $Y$. The probability density function (PDF) and cumulative density function (CDF) of $h_{XY}$ can be expressed as \cite[eq.2.21]{fadingchannels}
\begin{equation}\label{PdfPowerGain}
f_{h_{XY}}(x) = {{{{\beta_{XY}} ^{m_{XY}}}} \over {\Gamma (m_{XY})}}{x^{{m_{XY}} - 1}}\exp ( - {\beta_{XY}} x),
\end{equation}
\begin{equation}\label{CDFPowerGain}
{F_{h_{XY}}}\left( x \right) = {{\gamma \left( {m_{XY},\beta_{XY} x} \right)} \over {\Gamma \left( m_{XY} \right)}},
\end{equation}
where $\Gamma \left( z \right) = \int_0^\infty  {{t^{z - 1}}\exp \left( { - t} \right)} dt$ is the gamma function, $\beta_{XY}  = {m_{XY} \mathord{\left/
 {\vphantom {m {\overline h }}} \right.
 \kern-\nulldelimiterspace} {\Omega_{XY} }}$ and $\gamma \left( {v,x} \right) = \int_0^x {{t^{v - 1}}\exp \left( { - t} \right)dt} $ is the lower incomplete gamma function. Besides, we assume that all channels experience slow, independent and frequency flat fading that the channel power gain remains unchanged but can be different from one transmission block to another.

The time diagrams of HTT, HTC and AT protocols are depicted in Fig.\ref{fig:time diagram}. In all three protocols, the first $\tau T$ amount of each transmission block $T$ is allocated to the DL energy transfer. The harvested energy at the source and relay can thus be expressed as  \cite{harvest-then-transmit}
\begin{equation}\label{EatS}
{E_S} = \eta \tau T{P_A}{h_{AS}},
\end{equation}
\begin{equation}\label{EatR}
{E_R} = \eta \tau T{P_A}{h_{AR}},
\end{equation}
where $0 < \eta  < 1$ denotes the energy conversion efficiency and ${P_A}$ is the transmit power of the AP. We consider that the source and relay will exhaust the harvested energy in the subsequent UL information transmission phase. For simplicity, we assume a normalized transmission block, i.e., $T = 1$, in the rest of this paper. In the following, we will present the received signal-to-noise ratio (SNR) at the AP for the considered three protocols:
\begin{figure}
\centering
 \subfigure[HTT protocol]
  {\scalebox{0.25}{\includegraphics {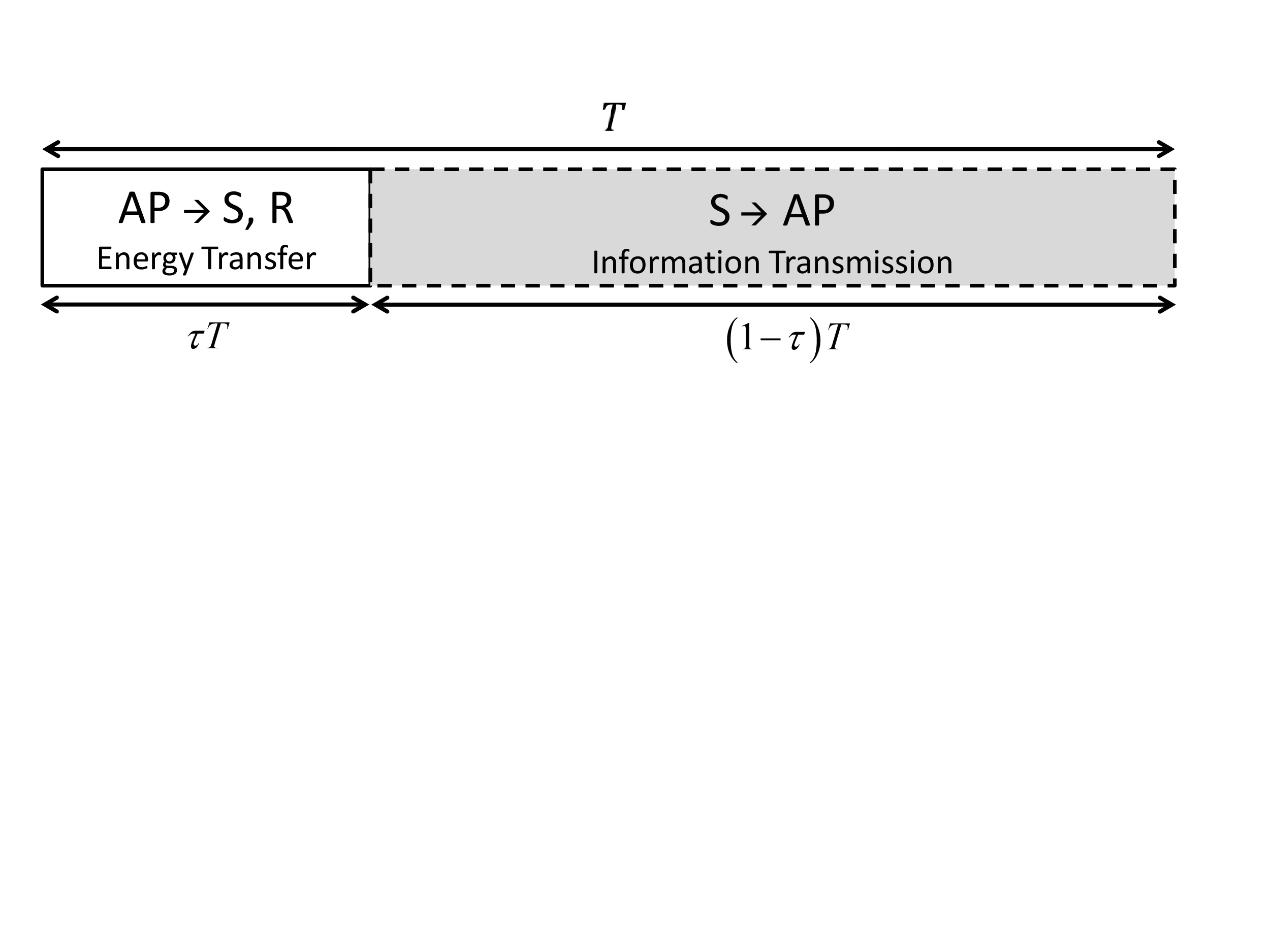}
  \label{fig:time diagram_a}}}
\hfil
 \subfigure[HTC protocol]
  {\scalebox{0.25}{\includegraphics {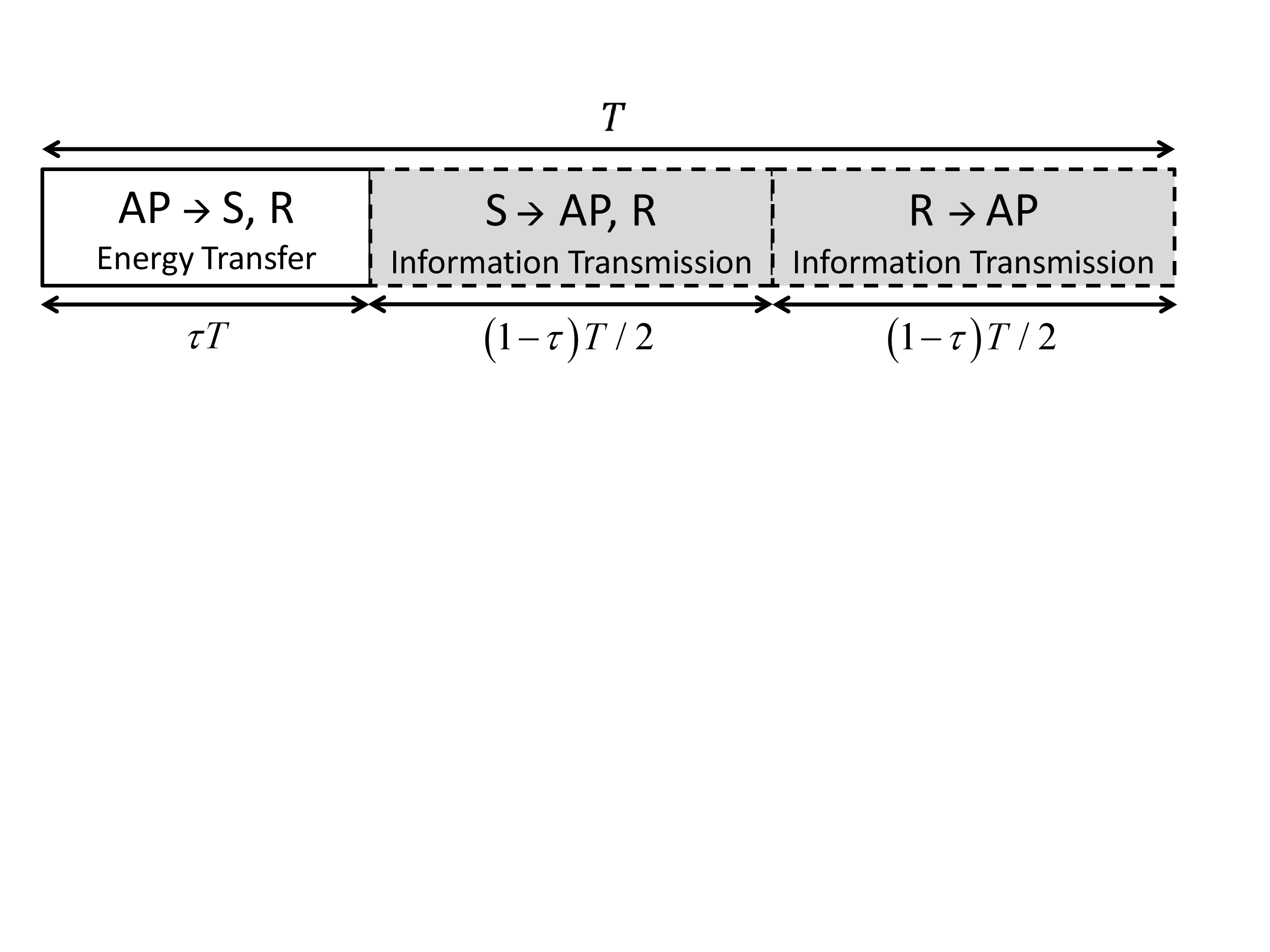}
\label{fig:time diagram_b}}}
\hfil
 \subfigure[AT protocol]
  {\scalebox{0.25}{\includegraphics {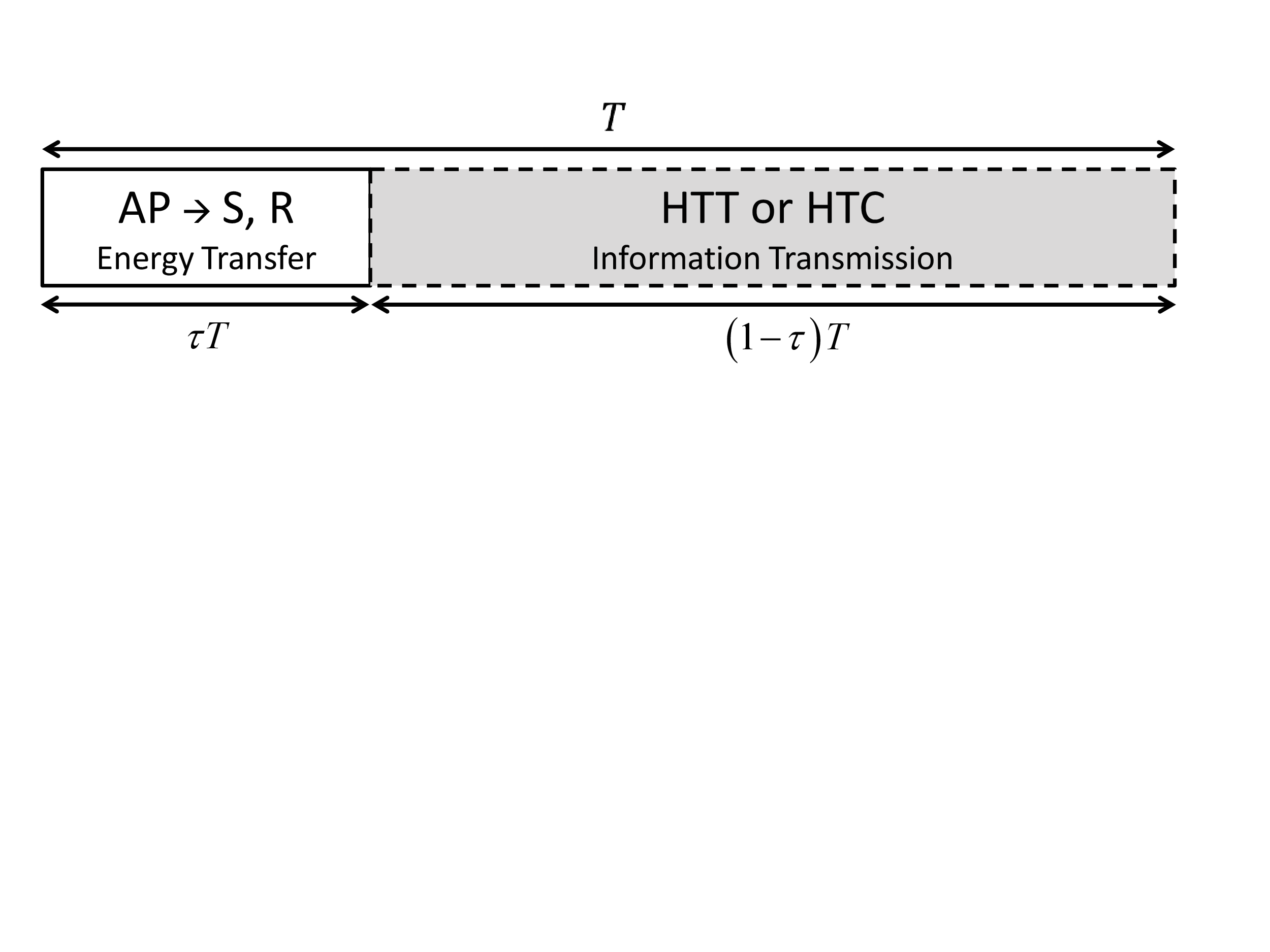}
\label{fig:time diagram_c}}}
\caption{Time diagrams of HTT, HTC and AT protocols.
\label{fig:time diagram}}
\end{figure}
%
%
\subsection{HTT protocol}
In the HTT protocol, after energy harvesting phase, the source will transmit information to the AP during the remaining time. The transmit power of the source is thus given by
\begin{equation}\
{P_{S,1}} = {E_S}/(1 - \tau)=\eta \tau {P_A}{h_{AS}}/(1 - \tau ).
\end{equation}
The received SNR at the AP can be expressed as
\begin{equation}\label{HTT-Ra1}
{\gamma _{A,1}} = {{{P_{S,1}}{h_{SA}}} \mathord{\left/
 {\vphantom {{{P_s}{h_{SA}}} {{N_0} = }}} \right.
 \kern-\nulldelimiterspace} {{N_0} = }}\mu_1 {h_{AS}}{h_{SA}},
\end{equation}
where ${N_0}$ is the power of the noise suffered by all receivers and
\begin{equation}\label{miu1}
\mu_1  = \eta ({P_A}/{N_0})\tau /(1 - \tau ).
\end{equation}
\subsection{HTC protocol}
In the HTC protocol, the remaining time after DL energy transfer is further divided into two time slots with equal length ${{\left( {1{\rm{ - }}\tau } \right)} \over 2}$. The first time slot is used for $S$ to transmit information to $R$ while the second time slot is used for $R$ to forward information to AP. Amplify-and-forward (AF) relaying is assumed to implement at the relay due to its simplicity and selection combining (SC) technique \cite{fadingchannels} is employed for information processing at the AP. Similarly, the transmit power of the source and relay can be expressed as
\begin{equation}\
{P_{S,2}} = {E_S}/[(1 - \tau)/2]=2\eta \tau {P_A}{h_{AS}}/(1 - \tau ),
\end{equation}
\begin{equation}\label{relaypower}
{P_{R,2}} = {E_R}/[(1 - \tau)/2]=2\eta \tau {P_A}{h_{AR}}/(1 - \tau ).
\end{equation}
The received SNR at the AP from the S-AP link can be written as
\begin{equation}\label{snr-s-ap}
{\gamma _{SA,2}} = {{{P_{S,2}}{h_{SA}}} \mathord{\left/
 {\vphantom {{{P_s}{h_{SA}}} {{N_0} = }}} \right.
 \kern-\nulldelimiterspace} {{N_0} = }}\mu_2 {h_{AS}}{h_{SA}},
\end{equation}
where
\begin{equation}\label{miu2}
\mu_2  = 2\eta ({P_A}/{N_0})\tau /(1 - \tau ).
\end{equation}

At the same time, the signal sent by the source can also be overheard by the relay. The relay will amplify and forward the received signal to AP using the power ${P_{R,2}}$ given in (\ref{relaypower}) with amplification factor $\beta {\rm{ = }}{1 \mathord{\left/
 {\vphantom {1 {\sqrt {{P_{S,2}}{h_{SR}} + {N_0}} }}} \right.
 \kern-\nulldelimiterspace} {\sqrt {{P_{S,2}}{h_{SR}} + {N_0}} }}$ \cite{amp-factor}. After some algebraic manipulations, we can express the received SNR  from S-R-A link as
\begin{equation}\label{Rsra}
{\gamma _{SRA,2}} = {{{\mu _2}{h_{AS}}{h_{SR}}{h_{AR}}{h_{RA}}} \over {{h_{AS}}{h_{SR}} + {h_{AR}}{h_{RA}} + 1/\mu_2}}.
\end{equation}
With the SC technique, the output SNR of the HTC protocol is given by
\begin{equation}\label{rA2}
{\gamma _{A,2}} = \max ({\gamma _{SA,2}},{\gamma _{SRA,2}}).
\end{equation}
\subsection{Adaptive transmission (AT) protocol}
In the proposed AT protocol, the AP can adaptively choose the HTT and HTC protocols. More specifically, according to its transmit power and the channel information $h_{AS}h_{SA}$, the AP can judge whether an outage will occur when the HTT protocol is implemented at the source. If an outage happens, the HTC protocol will be activated, otherwise the HTT protocol will be adopted.

We first characterize the threshold on the product $h_{AS}h_{SA}$, which can guarantee that no outage occurs for the HTT protocol. The mutual information between the source and AP can be written as
\begin{equation}\label{mutualHTT}
{I_{SA}^{HTT}} = {\log _2}\left( {1 + {\gamma _{A,1}}} \right).
\end{equation}
Let $R$ denote the transmission rate of the source. Then there is no outage for the HTT protocol when the following condition holds
\begin{equation}\label{criterion}
{I_{SA}^{HTT}}>R.
\end{equation}
Substituting (\ref{HTT-Ra1}) to (\ref{criterion}) and we have
\begin{equation}
{h_{AS}}{h_{SA}} > {{{\upsilon _1}} \over {{\mu _1}}},
\end{equation}
where $\upsilon_1  = {2^{R}} - 1$. Thus, we can summarize that in the AT protocol, HTT protocol is implemented if ${h_{AS}}{h_{SA}} > {{{\upsilon _1}} \over {{\mu _1}}}$ while HTC protocol is employed if ${h_{AS}}{h_{SA}} \le {{{\upsilon _1}} \over {{\mu _1}}}$. Then, the output SNR at the AP with this AT protocol can be characterized by the following two cases:
\subsubsection{${h_{AS}}{h_{SA}} > {{{\upsilon _1}} \over {{\mu _1}}}$}
The HTT protocol is implemented and the received SNR at the AP is give by
\begin{equation}\label{rA3_HTT}
{\gamma _{A,3}} = {\gamma _{A,1}}.
\end{equation}
Note that no outage will occur in this case.
\subsubsection{${h_{AS}}{h_{SA}} \le {{{\upsilon _1}} \over {{\mu _1}}}$}
The HTC protocol is adopted in this case. We thus can express the received SNR at the AP as
\begin{equation}\label{rA3}
{\gamma _{A,3}} = {\gamma _{A,2}}.
\end{equation}
\section{Throughput Analysis}
In this paper, we consider a delay-limited transmission mode \cite{Delay-limitedthroughput}, where the throughput can be determined by evaluating the outage probability with a fixed transmission rate. To proceed, we first characterize the outage probabilities of HTT, HTC and AT protocols over Nakagami-m fading channels and have the following propositions.
\begin{proposition}\label{propHTT}
The expression of the outage probability of the HTT protocol over Nakagami-m fading channels can be expressed as
\begin{equation}\label{pouthtt}
P_{out}^{HTT}= 1 - {S_{{m_{AS}},{m_{SA}}}}\left( {{{{\beta _{AS}}{\beta _{SA}}{\upsilon _1}} \over {{\mu _1}}}} \right),
\end{equation}
where
\begin{equation}
{S_{{m_1},{m_2}}}\left( x \right) = {2 \over {\Gamma \left( {{m_2}} \right)}}\sum\limits_{i = 0}^{{m_1} - 1} {{{{x^{{{{m_2} + i} \over 2}}}} \over {i!}}} {K_{{m_2} - i}}\left( {2\sqrt x } \right),
\end{equation}
and ${K_\upsilon}\left( z \right)$ is the modified Bessel function of the second kind with order $\upsilon $ \cite[eqn.8.407]{Tableofintegral}.
\end{proposition}
\begin{proof}
Based on the analysis in Sec. II-C, the outage probability of HTT protocol can be expressed as
\begin{equation}\label{20}
\begin{split}
P_{out}^{HTT}&=\Pr\left( {{h_{AS}}{h_{SA}} < {\upsilon_1  \over \mu_1 }} \right)\\
&= \int_0^\infty  \Pr \left( {{h_{AS}} < {\upsilon_1  \over {\mu_1 y}}} \right){f_{{h_{SA}}}}\left( y \right)dy \\
&= \int_0^\infty  {{{\gamma ({m_{AS}},{{{\beta _{AS}}\upsilon_1 } \over {\mu_1 y}})} \over {\Gamma ({m_{AS}})}}} {f_{{h_{SA}}}}\left( y \right)dy.\\
\end{split}
\end{equation}
To achieve a closed-from expression for the last integral in (\ref{20}), we apply the following formula \cite[eqn.(8.352.6)]{Tableofintegral}
\begin{equation}
{{\gamma (m,x)} \over {\Gamma (m)}} = 1 - \exp ( - x)\sum\limits_{i = 0}^{m - 1} {{{{x^i}} \over {i!}}}.
\end{equation}
Then, we can further evaluate the integral in (\ref{20}) as
\begin{equation}\label{P1}
\begin{split}
  P_{out}^{HTT} &= \int_0^\infty  \left( {1 - \exp ( - {{{\beta _{AS}}{\upsilon _1}} \over {{\mu _1}y}})\sum\limits_{i = 0}^{{m_{AS}} - 1} {{{{{\left( {{{{\beta _{AS}}{\upsilon _1}} \over {{\mu _1}}}} \right)}^i}} \over {i!{y^i}}}} } \right)\\
  & \quad \times {f_{{h_{SA}}}}\left( y \right)dy   \\
  &  = 1 - {{{\beta _{SA}}^{{m_{SA}}}} \over {\Gamma ({m_{SA}})}}\sum\limits_{i = 0}^{{m_{AS}} - 1} {{{{{\left( {{{{\beta _{AS}}\upsilon_1 } \over \mu_1 }} \right)}^i}} \over {i!}}} \times \\
  &\quad \int_0^\infty  {{y^{{m_{SA}} - i - 1}}\exp \left( { - {{{\beta _{AS}}\upsilon_1 } \over {\mu_1 y}} - {\beta _{SA}}y} \right)dy} \\
  &  = 1 - {S_{{m_{AS}},{m_{SA}}}}\left( {{{{\beta _{AS}}{\beta _{SA}}{\upsilon _1}} \over {{\mu _1}}}} \right), \\
\end{split}
\end{equation}
where the last integral is solved by \cite[eqn.(3.471-9)]{Tableofintegral}.
\end{proof}

Note that the closed-form expressions for the exact outage probabilities of the HTC and AT protocols are analytically intractable due to the complicated structure of the received SNR. Motivated by this, we apply the approximation used in \cite{chen2014harvest} to approximate the received SNRs and derive approximate expressions for the outage probabilities of the HTC and AT protocols given in the following two propositions.
\begin{proposition}\label{propHTC}
An approximate closed form expression for the outage probability of the HTC protocol over Nakagami-m fading channels can be expressed as
\begin{equation}\label{HTCpout}
\begin{split}
  & {P_{out}^{HTC}} \approx   \\
  & 1 - {S_{{m_{AS}},{m_{SA}}}}\left( {{{{\beta _{AS}}{\beta _{SA}}\upsilon_2 } \over \mu_2 }} \right) - {S_{{m_{AR}},{m_{RA}}}}\left( {{{{\beta _{AR}}{\beta _{RA}}\upsilon_2 } \over \mu_2 }} \right) \times   \\
  & \left[ {{S_{{m_{SR}},{m_{AS}}}}\left( {{{{\beta _{SR}}{\beta _{AS}}\upsilon_2 } \over \mu_2 }} \right) - \Phi \left( {{\beta _{SA}}{\upsilon_2  \over \mu_2 },{\beta _{SR}}{\upsilon_2  \over \mu_2 }} \right)} \right], \\
  \end{split}
\end{equation}
where
\begin{equation}
\begin{split}
  & \Phi \left( {{x_1},{x_2}} \right) =   \\
  & {2 \over {\Gamma \left( {{m_{AS}}} \right)}}\sum\limits_{i = 0}^{{m_{SA}} - 1} {\sum\limits_{j = 0}^{{m_{SR}} - 1} {\left[ {{{{{\left( {{x_1}} \right)}^i}} \over {i!}}} \right.} {{{{\left( {{x_2}} \right)}^j}} \over {j!}}} {\left( {{x_1} + {x_2}} \right)^{{{{m_{AS}} - i - j} \over 2}}}  \\
  & \left. { \times {{\left( {{\beta _{AS}}} \right)}^{{{{m_{AS}} + i + j} \over 2}}}{K_{{m_{AS}} - i - j}}\left( {2\sqrt {\left( {{x_1} + {x_2}} \right){\beta _{AS}}} } \right)} \right], \\
\end{split}
\end{equation}
and ${\upsilon _2}{\rm{ = }}{2^{2R}} - 1$.
\end{proposition}
\begin{proof}
See Appendix A.
\end{proof}
\begin{proposition}\label{propAT}
An approximate closed form expression for the outage probability of the proposed AT protocol over Nakagami-m fading channels is given by
\begin{equation}
\begin{split}
{P_{out}^{AT}} \approx & 1{\rm{ - }}{S_{{m_{AS}},{m_{SA}}}}\left( {{{{\beta _{AS}}{\beta _{SA}}{\upsilon _1}} \over {{\mu _1}}}} \right) \\
&- {S_{{m_{AR}},{m_{RA}}}}\left( {{{{\beta _{AR}}{\beta _{RA}}{\upsilon _2}} \over {{\mu _2}}}} \right)\times \\
& \left[ {{S_{{m_{SR}},{m_{AS}}}}\left( {{{{\beta _{SR}}{\beta _{AS}}{\upsilon _2}} \over {{\mu _2}}}} \right) - \Phi \left( {{\beta _{SA}}{{{\upsilon _1}} \over {{\mu _1}}},{\beta _{SR}}{{{\upsilon _2}} \over {{\mu _2}}}} \right)} \right].\\
\end{split}
\end{equation}
\end{proposition}
\begin{proof}
See Appendix B.
\end{proof}

Given a fixed data rate $R$ and the approximate outage probability, we can express the average throughput of the HTT, HTC and AT protocols as
\begin{equation}\label{throughput}
{{\Psi }_{Z}} \approx R\left( {1 - {P_{out}^{Z}}} \right)\left( {1 - \tau } \right),
\end{equation}
where $Z \in \left\{ {HTT,HTC,AT} \right\}$.
\begin{remark}
From the above analysis, it can be concluded that the outage probability is inversely proportional to the received SNR ${{\gamma _{A,3}}}$ and ${{\mu _1}}$ for the AT protocol. Furthermore, the received SNR given in (\ref{rA3}) is proportional to ${{\mu _2}}$ defined in (\ref{miu2}). Both ${{\mu _1}}$ and ${{\mu _2}}$ are proportional to the allocated time $\tau$. Thus, the outage probability of the AT protocol is inversely proportional to $\tau$. This observation is understandable since the more time allocated for DL energy transfer, the more energy the source and relay can harvest. This guarantees a higher transmit power, which results in a lower outage probability. However, the average throughput is also proportional to the information transmission time $1-\tau$. Jointly considering the above analysis, we can deduce that there must be an optimal value for the DL energy transfer time $\tau $ between 0 and 1 such that the average throughput is maximized. In the future work, this could be an interesting work to derive. Same results can be obtained for HTT and HTC protocols.
\end{remark}
\section{Numerical Results}
In this section, we present some numerical results to illustrate and validate the above theoretical analyses. We consider the scenario that ${m_{XY}}= m$ for simplicity but their average powers $\Omega _{XY}$ may differ from each other. In order to capture the effect of path-loss, we use the model that ${\Omega _{XY}} = {10^{ - 3}}{\left( {{d_{XY}}} \right)^{ - \alpha }}$, where ${d_{XY}}$ denotes the distance between nodes $X$ and $Y$, and $\alpha  \in \left[ {2,5} \right]$ is the path-loss factor \cite{path-loss}. We consider a linear topology that the relay is located on a straight line between the source and hybrid AP. In all following simulations, we set the distance between the AP and source ${d_{AS}} = 10m$, the path-loss factor $\alpha  = 2$, the noise power ${N_0} =  - 80dBm$ and the energy conversion efficiency $\eta  = 0.5$.

We first compare the analytical throughput derived in above sections with the Monte Carlo simulation results. The throughput of HTT, HTC and the proposed AT protocols versus the AP transmit power is shown in Fig.\ref{fig:montsimu1} and Fig.\ref{fig:montsimu2}. We can see that the simulation results agree with the exact expression of throughput for the HTT protocol and the derived approximate expressions of throughput for the HTC and AT protocols become very tight at medium and high transmit power conditions. This numerical results validate our theoretical results presented in Proposition 1-3. Also, it can be observed that the proposed AT protocol outperforms the HTC and HTT protocols for all transmit powers.
\begin{figure}
\centering \scalebox{0.5}{\includegraphics{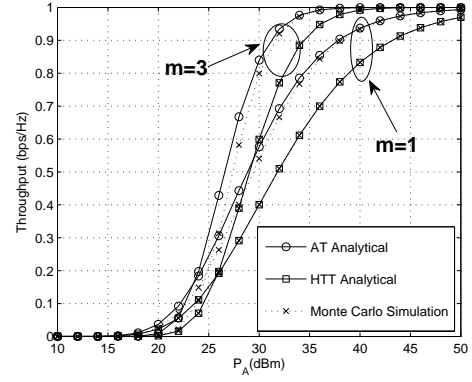}}
\caption{Average throughput of AT and HTT protocols versus the AP transmit power for $m=1$ and $m=3$, where ${d_{AR}}=5$, $\tau=0.5$ and $R=2$.\label{fig:montsimu1}}
\end{figure}
\begin{figure}
\centering \scalebox{0.5}{\includegraphics{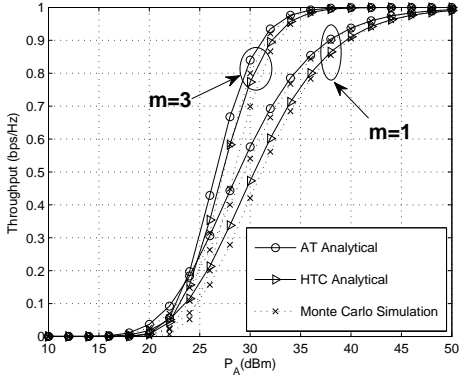}}
\caption{Average throughput of AT and HTC protocols versus the AP transmit power for $m=1$ and $m=3$, where ${d_{AR}}=5$, $\tau=0.5$ and $R=2$.\label{fig:montsimu2}}
\end{figure}

We will only plot the analytical results in the remaining figures since the theoretical analyses agree with the simulations when the AP transmit power is high enough. In Fig.\ref{fig:TdifferentSNR}, we plot the throughput curves versus DL energy transfer time $\tau$ with different transmit power ${P_A} = 35,40dBm$ respectively. It validates our analysis in Remark 1 that there exists an optimal $\tau$ that maximizes the average throughput in all cases. But, the optimal DL energy transfer time may be not the same for different protocols. In these three protocols, the optimal energy transfer time of the proposed AT is the smallest, which means that the AT protocol consumes the least energy at the AP. Moreover, the higher the transmit power at the AP, the smaller the optimal value of $\tau$. This is because that the source can harvest the same amount of energy with shorter time when the transmit power increases. Hence, the optimal $\tau$ will shift to the left as the value of $P_A$ increases. We can also see from Fig.\ref{fig:TdifferentSNR} that with the optimal DL energy transfer time, the proposed AT protocol can introduce a significant performance gain compared with both HTT and HTC protocols.
\begin{figure}
\centering \scalebox{0.5}{\includegraphics{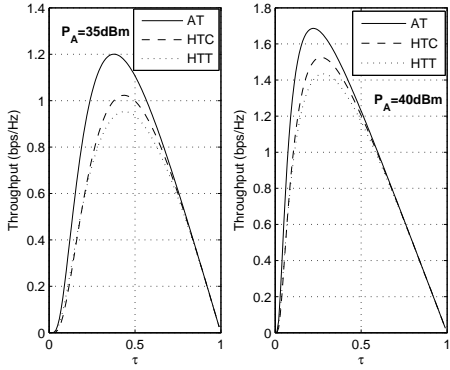}}
\caption{Average throughput of the three protocols versus $\tau $ for different $P_A$, where $R=2.5$, $d_{AR}=5$, and $m=2$. \label{fig:TdifferentSNR}}
\end{figure}

\begin{figure}
\centering \scalebox{0.5}{\includegraphics{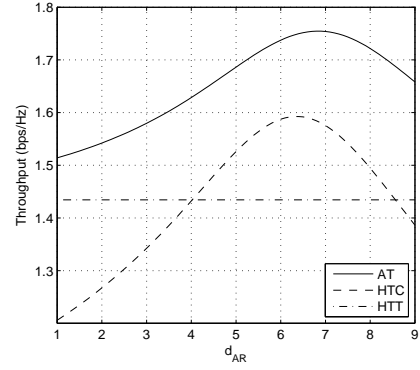}}
\caption{Average throughput of the three protocols versus $d_{AR}$, where $P_A=40dBm$, $m=2$, $R =2.5$ and $\tau$ is at its optimal value. \label{fig:distance}}
\end{figure}

Fig.~\ref{fig:distance} illustrates the effect of relay location on the average throughput in which the throughput curves are plotted versus ${d_{AR}}$ with the optimal values of $\tau$. Note that the optimal value of $\tau$ can be easily found by one-dimension exhaustive search. Since HTT protocol does not have the external relay, the throughput is plotted as a straight line for comparison. We can see that the throughput for the proposed AT protocol is larger than that of the HTC and HTT protocols in all distance cases. Furthermore, the relay should be placed close to the source in order to achieve maximum throughput for HTC and AT protocols. Finally, jointly considering Figs. 3-6, we can conclude that the proposed AT protocol is superior to HTT and HTC protocols under all considered cases.

%
\section{Conclusions}
In this paper, we proposed an adaptive transmission (AT) protocol for wireless-powered cooperative communication networks. We derived the outage probability and average throughput of the proposed AT protocol over Nakagami-m fading channels. In order to compare the AT protocol with the existing HTT and HTC protocols, we also analyzed the performances of HTT and HTC  protocols over Nakagami-m fading channels. Numerical results validated our analyses and shown that the proposed AT protocol outperforms the HTC and HTT protocols under all simulations.
\appendices
\section{Proof of proposition 2}
For the HTC protocol, the mutual information between the source and hybrid AP is given by
\begin{equation}\label{htc-mutual}
{I_{SA}^{HTC}} = {1 \over 2}{\log _2}\left( {1 + {\gamma _{A,2}}} \right).
\end{equation}
The outage probability can thus be expressed as
\begin{equation}\label{outageHTC}
\begin{split}
{P_{out}^{HTC}} &= \Pr\left( {{I_{SA}^{HTC}} < R} \right)\\
&= \Pr({\gamma _{A,2}} < \upsilon_2 ) = \Pr({\gamma _{SA,2}} < \upsilon_2 ,{\gamma _{SRA,2}} < \upsilon_2 ).
\end{split}
\end{equation}
We make the following approximation to the expression of ${\gamma _{SRA,2}}$ given in (\ref{Rsra}) in order to calculate the outage probability \cite{chen2014harvest}
\begin{equation}\label{approRSRA}
\begin{split}
{\gamma _{SRA,2}} &\approx {{\mu_2 {h_{AS}}{h_{SR}}{h_{AR}}{h_{RA}}} \over {{h_{AS}}{h_{SR}} + {h_{AR}}{h_{RA}}}} \\
&\approx \mu_2 \min \left( {{h_{AS}}{h_{SR}},{h_{AR}}{h_{RA}}} \right).\\
\end{split}
\end{equation}
Based on the approximation in (\ref{approRSRA}), we can obtain an approximate expression of outage probability for the HTC protocol
\begin{equation}\label{approPout}
\begin{split}
&{P_{out}^{HTC}}\approx \Pr\left( {{h_{AS}}{h_{SA}} < {\upsilon_2  \over \mu_2 },\min \left( {{h_{AS}}{h_{SR}},{h_{AR}}{h_{RA}}} \right) < {\upsilon_2  \over \mu_2 }} \right) \\
&= \Pr\left( {{h_{AS}}{h_{SA}} < {\upsilon_2  \over \mu_2 }} \right) - \Pr\left( {{h_{AR}}{h_{RA}} > {\upsilon_2  \over \mu_2 }} \right) \times \\
& \quad \Pr\left( {{h_{AS}}{h_{SA}} < {\upsilon_2  \over \mu_2 },{h_{AS}}{h_{SR}} > {\upsilon_2  \over \mu_2 }} \right),
\end{split}
\end{equation}
where the second equality follows the fact that $\Pr \left( {A,B} \right) = \Pr \left( A \right) - \Pr \left( {A,\overline B } \right)$ \cite[pp.11]{prob.P11-1.7}.
Now we calculate the three probability terms in the last step of (\ref{approPout}). Similarly as (\ref{pouthtt}), the first and second terms can be calculated as
\begin{equation}\label{P1}
 \Pr\left( {{h_{AS}}{h_{SA}} < {{{\upsilon _2}} \over {{\mu _2}}}} \right)= 1 - {S_{{m_{AS}},{m_{SA}}}}\left( {{{{\beta _{AS}}{\beta _{SA}}{\upsilon _2}} \over {{\mu _2}}}} \right),
\end{equation}
\begin{equation}\label{P2}
\Pr\left( {{h_{AR}}{h_{RA}} > {\upsilon_2  \over \mu_2 }} \right)= {S_{{m_{AR}},{m_{RA}}}}\left( {{{{\beta _{AR}}{\beta _{RA}}{\upsilon _2}} \over {{\mu _2}}}} \right).
\end{equation}
For the last term, we have
\begin{equation}\label{P3}
\begin{split}
&\Pr\left( {{h_{AS}}{h_{SA}} < {\upsilon_2  \over \mu_2 },{h_{AS}}{h_{SR}} > {\upsilon_2  \over \mu_2 }} \right) \\
&= \int_0^\infty  {P({h_{SA}}}  < {\upsilon_2  \over {\mu_2 y}})P({h_{SR}} > {\upsilon_2  \over {\mu_2 y}}){f_{{h_{AS}}}}\left( y \right)dy \\
&= \underbrace {\int_0^\infty  {P({h_{SR}} > {\upsilon_2  \over {\mu_2 y}})} {f_{{h_{AS}}}}\left( y \right)dy}_{{I_1}} - \\
&\underbrace {\int_0^\infty  {P({h_{SA}}}  > {\upsilon_2  \over {\mu_2 y}})P({h_{SR}} > {\upsilon_2  \over {\mu_2 y}}){f_{{h_{AS}}}}\left( y \right)dy}_{{I_2}}.
\end{split}
\end{equation}
Similarly as (\ref{pouthtt}), ${I_1}$ can be calculated as
\begin{equation}
{I_1} = {S_{{m_{SR}},{m_{AS}}}}\left( {{{{\beta _{SR}}{\beta _{AS}}{\upsilon _2}} \over {{\mu _2}}}} \right).
\end{equation}
${I_2}$ can be evaluated as
\begin{equation}
\begin{split}
  {I_2} &= {{{\beta _{AS}}^{{m_{AS}}}} \over {\Gamma ({m_{AS}})}}\sum\limits_{i = 0}^{{m_{SA}} - 1} {\sum\limits_{j = 0}^{{m_{SR}} - 1} {{{{{\left( {{{{\beta _{SA}}{\upsilon _2}} \over {{\mu _2}}}} \right)}^i}} \over {i!}}{{{{\left( {{{{\beta _{SR}}{\upsilon _2}} \over {{\mu _2}}}} \right)}^j}} \over {j!}}} }  \times   \\
  & \int_0^\infty  {{y^{{m_{AS}} - i - j - 1}}\exp \left( { - {{\left( {{\beta _{SA}} + {\beta _{SR}}} \right){{{\upsilon _2}} \over {{\mu _2}}}} \over y} - {\beta _{AS}}y} \right)dy}   \cr
  &  = \Phi \left( {{\beta _{SA}}{{{\upsilon _2}} \over {{\mu _2}}},{\beta _{SR}}{{{\upsilon _2}} \over {{\mu _2}}}} \right), \\
\end{split}
\end{equation}
where the last integral is solved by \cite[eqn.(3.471-9)]{Tableofintegral}. Substituting (\ref{P1}), (\ref{P2}), (\ref{P3}) to (\ref{approPout}), we obtain the desired result in Proposition \ref{propHTC}. This completes the proof.

\section{Proof of proposition 3}
In the AT protocol, outage will only happen when the HTC protocol is employed under the channel information condition ${h_{AS}}{h_{SA}} \le {{{v_1}} \over {{\mu _1}}}$ and the corresponding mutual information is less than the required rate $R$ at the same time. Mathematically, the outage probability of the AT protocol can thus be expressed as
\begin{equation}\label{outageAT}
\begin{split}
 {P_{out}^{AT}} &= \Pr \left( {{I_{SA}^{HTC}} < R,{h_{AS}}{h_{SA}} \le {{{v_1}} \over {{\mu _1}}}} \right)\\
 &\approx \Pr\left( {{h_{AS}}{h_{SA}} < {\upsilon_2  \over \mu_2 },\min \left( {{h_{AS}}{h_{SR}},{h_{AR}}{h_{RA}}} \right) < {\upsilon_2  \over \mu_2 }} \right. ,\\
& \quad \left. {{h_{AS}}{h_{SA}} < {{{v_1}} \over {{\mu _1}}}} \right).\\
\end{split}
\end{equation}
To further simplify the above formula, we compare the terms ${{{v_1}} \over {{\mu _1}}}$ and ${{{v_2}} \over {{\mu _2}}}$ and obtain
\begin{equation}
\frac{{{\upsilon _2}}}{{{\mu _2}}} - \frac{{{\upsilon _1}}}{{{\mu _1}}} = \frac{1}{{{\mu _1}}}\left( {{2^{2R - 1}} - {2^R} + \frac{1}{2}} \right) > 0,
\end{equation}
where the last inequality follows since ${{2^{2R - 1}} - {2^R}} > -\frac{1}{2}$ holds when $R>0$. Thus, ${{{v_1}} \over {{\mu _1}}} < {{{v_2}} \over {{\mu _2}}}$ for all $R>0$ and the outage probability of the AT protocol can be further evaluated as
\begin{equation}\label{poutAT}
\begin{split}
{P_{out}^{AT}} &\approx \Pr \left( {\min \left( {{h_{AS}}{h_{SR}},{h_{AR}}{h_{RA}}} \right) < {{{v_2}} \over {{\mu _2}}},{h_{AS}}{h_{SA}} < {{{v_1}} \over {{\mu _1}}}} \right)  \\
  &  = \Pr \left( {{h_{AS}}{h_{SA}}< {{{v_1}} \over {{\mu _1}}}} \right) - \Pr \left( {{h_{AR}}{h_{RA}} > {{{v_2}} \over {{\mu _2}}}} \right) \times \\
  &\quad \Pr \left( {{h_{AS}}{h_{SR}} > {{{v_2}} \over {{\mu _2}}},{h_{AS}}{h_{SA}} < {{{v_1}} \over {{\mu _1}}}} \right). \\
\end{split}
\end{equation}
With the similar methods used in Appendix A, we can evaluate the last three probability terms in (\ref{poutAT}) and obtain the desired expression given in Proposition \ref{propAT}.

\ifCLASSOPTIONcaptionsoff
  \newpage
\fi

\bibliographystyle{IEEEtran}
\bibliography{References}

\end{document}